\documentclass[a4paper,onecolumn,aps,nofootinbib,superscriptaddress,tightenlines,11pt,notitlepage,pra]{revtex4-1}
\usepackage{dsfont}
\usepackage{amsmath}
\usepackage{amssymb}
\usepackage{amsthm}
\usepackage{color}
\usepackage{graphicx}
\usepackage{mathtools}

\usepackage[ colorlinks = true,
             linkcolor = blue,
             urlcolor  = blue,
             citecolor = red,
             anchorcolor = blue,
]{hyperref}

\DeclareMathOperator{\tr}{tr}

\newcommand{\id}{1}
\newcommand{\M}{\mathbb{M}}
\newcommand{\Proj}{\mathbb{P}}
\newcommand{\cX}{\mathcal{X}}
\newcommand{\cI}{\mathcal{I}}

\newcommand{\rec}{\mathrm{rec}}
\newcommand{\cP}{\mathcal{P}}
\newcommand{\cS}{\mathcal{S}}

\theoremstyle{plain}
\newtheorem{lemma}{Lemma}
\newtheorem{prop}[lemma]{Proposition}
\newtheorem{theorem}[lemma]{Theorem}

\theoremstyle{definition}

\begin{document}

\author{Mario Berta}
\affiliation{Department of Computing, Imperial College London}
\affiliation{Institute for Quantum Information and Matter, California Institute of Technology}
%, Pasadena, California 91125, USA}
%\email{berta@caltech.edu}

\author{Omar Fawzi}
% [MT] Is this really an official affiliation? If its just a thank you note, it should go in acknowledgements.
%\affiliation{Department of Computing and Mathematical Sciences, California Institute of Technology}
%, Pasadena, CA 91125, USA}
% [MT] The following line looks more like something that should go in acknowledgements?}
\affiliation{Laboratoire de l'Informatique du Parall\'{e}lisme, \'{E}cole Normale Sup\'{e}rieure de Lyon}
%, Lyon, 69007, France}
%\thanks{LIP UMR 5668 LIP - ENS Lyon - CNRS - UCBL - INRIA}
%\email{omar.fawzi@ens-lyon.fr}

\author{Marco Tomamichel}
\affiliation{Centre for Quantum Software and Information, University of Technology Sydney}
%\affiliation{School of Physics, The University of Sydney}
%, Sydney, NSW 2006, Australia}
%\email{marco.tomamichel@sydney.edu.au}

%\title{\Large Variational Expressions for Measured Relative Entropy}
%\title{\Large On Variational Expressions for Measured Relative Entropy}

%%%%%%%%%%%%%%%%%%%%%%%%%%%%%%%%%%%%%%%%%%%%

\title{\Large On Variational Expressions for Quantum Relative Entropies}

\begin{abstract}
Distance measures between quantum states like the trace distance and the fidelity can naturally be defined by optimizing a classical distance measure over all measurement statistics that can be obtained from the respective quantum states. In contrast, Petz showed that the measured relative entropy, defined as a maximization of the Kullback-Leibler divergence over projective measurement statistics, is strictly smaller than Umegaki's quantum relative entropy whenever the states do not commute. We extend this result in two ways. First, we show that Petz' conclusion remains true if we allow general positive operator valued measures. Second, we extend the result to R\'enyi relative entropies and show that for non-commuting states the sandwiched R\'enyi relative entropy is strictly larger than the measured R\'enyi relative entropy for $\alpha \in (\frac12, \infty)$, and strictly smaller for $\alpha \in [0,\frac12)$. The latter statement provides counterexamples for the data-processing inequality of the sandwiched R\'enyi relative entropy for $\alpha < \frac12$. Our main tool is a new variational expression for the measured R\'enyi relative entropy, which we further exploit to show that certain lower bounds on quantum conditional mutual information are superadditive.\\


\end{abstract}

\maketitle

%%%%%%%%%%%%%%%%%%%%%%%%%%%%%%%%%%%%%%%%%%%%

\section{Measured Relative Entropy}

The relative entropy is the basic concept underlying various information measures like entropy, conditional entropy and mutual information. A thorough understanding of its quantum generalization is thus of preeminent importance in quantum information theory. We start by considering measured relative entropy, which is defined as a maximization of the Kullback-Leibler divergence over all measurement statistics that are attainable from two quantum states.

%They show that the measured relative entropy is strictly smaller than Umegaki's quantum relative entropy whenever the two quantum states do not commute. Here we extend a variational formula for the measured relative entropy by Petz to include the supremum over all finite outcome POVMs, from which it follows that the strict inequality to the quantum relative entropy persists even for this more general form. We then introduce the measured R\'enyi divergences between two quantum states and give two different variational characterizations thereof. These show that the measured quantities can be achieved by rank-1 projective measurements instead of general finite outcome POVMs. Furthermore, comparing these variational formulas with similar expressions by Frank and Lieb for the sandwiched R\'enyi divergences, we find that the measured R\'enyi divergences are equal to the corresponding sandwiched R\'enyi divergences if and only if the two quantum states commute (except for the special cases with R\'enyi parameters $1/2$ and $\infty$ where they are always equal). This also leads to analytical counterexamples for the data-processing inequality for the sandwiched R\'enyi divergences with R\'enyi parameters smaller than $1/2$, whenever the two quantum states do not commute. 

For a positive measure $Q$ on a finite set $\cX$ and a probability measure $P$ on $\cX$ that is absolutely continuous with respect to $Q$, denoted $P \ll Q$, the relative entropy or Kullback-Leibler divergence~\cite{kullback51} is defined as
\begin{align}
D(P\|Q) &:= \sum_{x \in \cX} P(x) \log \frac{P(x)}{Q(x)} \,,%\int_{\cX} P\big(\d x\big) \log\frac{\d P}{\d Q}\big(x \big)
\end{align}
where we understand $P(x) \log \frac{P(x)}{Q(x)} = 0$ whenever $P(x) = 0$. By continuity we define it as $+\infty$ if $P \not\ll Q$. (We use $\log$ to denote the natural logarithm.)

To extend this concept to quantum systems, Donald~\cite{donald86} as well as Hiai and Petz~\cite{hiai91} studied measured relative entropy. In the following we restrict ourselves to a $d$-dimensional Hilbert space for some $d \in \mathbb{N}$. Let us denote the set of positive semidefinite operators acting on this space by $\cP$ and the subset of density operators (with unit trace) by $\cS$. For a density operator $\rho \in \cS$ and $\sigma \in \cP$, we define two variants of measured relative entropy. The general \emph{measured relative entropy} is defined as
\begin{align}
D^{\M}(\rho\|\sigma):=\sup_{(\cX, M)  } D\big( P_{\rho,M} \big\| P_{\sigma,M} \big) \,,
\end{align}
where the optimization is over finite sets $\cX$ and positive operator valued measures (POVMs) $M$ on $\cX$. (More formally, $M$ is a map from $\cX$ to positive semidefinite operators and satisfies $\sum_{x \in \cX} M(x) = 1$, whereas $P_{\rho,M}$ is a measure on $\cX$ defined via the relation $P_{\rho,M}(x) = \tr[ M(x) \rho ]$ for any $x \in \cX$.) At first sight this definition seems cumbersome because we cannot restrict the size of $\cX$ that we optimize over. Therefore, following the works~\cite{donald86,petz86b,hiai91}, let us also consider the following \emph{projectively measured relative entropy}, which is defined as
\begin{align}\label{eq:relative_measured}
D^{\Proj}(\rho\|\sigma) := \sup_{ \{ P_i \}_{i=1}^d} \left\{ \sum_{i=1}^d \tr [P_i \rho] \log \frac{\tr [P_i \rho]}{\tr [P_i \sigma]} \right\}\,,
\end{align}
where the maximization is over all sets of mutually orthogonal projectors $\{ P_i \}_{i=1}^d$ and we spelled out the Kullback-Leibler divergence for discrete measures. Note that without loss of generality we can assume that these projectors are rank-$1$ as any course graining of the measurement outcomes can only reduce the relative entropy due to its data-processing inequality~\cite{lindblad75,uhlmann77}. Furthermore, the quantity is finite and the supremum is achieved whenever $\rho \ll \sigma$, which here denotes that the support of $\rho$ is contained in the support of $\sigma$. (To verify this, recall that the rank-$1$ projectors form a compact set and the divergence is lower semi-continuous.)

The first of the following two variational expressions for the (projectively) measured relative entropy is due to Petz~\cite{petzbook08}. Note that the second objective function is concave in $\omega$ so that the optimization problem has a particularly appealing form.

\begin{lemma}\label{lm:var_rel}
For $\rho \in \cS$ and $\sigma \in \cP$ non-zero, the following identities hold:
\begin{align}\label{eq:lemma3}
D^{\Proj}(\rho\|\sigma)=\sup_{\omega > 0}\ \tr[\rho \log \omega] - \log \tr[\sigma \omega]=\sup_{\omega > 0} \tr[\rho \log \omega]+1-\tr[\sigma \omega]\,.
\end{align}
Moreover, the suprema are achieved when $\rho$ and $\sigma$ are positive definite operators, i.e. $\rho > 0$ and $\sigma > 0$.
\end{lemma}

\begin{proof}
If $\rho \not\ll \sigma$ then the two expressions in the suprema of~\eqref{eq:lemma3} are unbounded, as expected. We now assume that $\rho \ll \sigma$. Let us consider the second expression in~\eqref{eq:lemma3}. We write the supremum over $\omega > 0$ as two suprema over $\{ P_i \}_{i=1}^d$ and $\{ \omega_i \}_{i=1}^d$, where $\omega_i > 0$ are the eigenvalues of $\omega$ corresponding to the eigenvectors given by rank-$1$ projectors $P_i$. Using the fact that $\tr[\rho] = 1$, we find
\begin{align}\label{eq:supsup1}
\sup_{\omega > 0} \tr[\rho \log \omega] + 1 - \tr[\sigma \omega] &= \sup_{ \{ P_i \}_{i=1}^d} \sup_{\{ \omega_i \}_{i=1}^d} \sum_{i=1}^d \tr [P_i \rho] (\log  \omega_i + 1) - \tr[P_i \sigma] \omega_i\,.
\end{align}
For $i \in [d]$ such that $\tr[P_i \sigma] = 0$, we also have $\tr[P_i \rho] = 0$, and thus the corresponding term is zero. When $\tr[P_i \sigma] > 0$, let us first consider $\tr[P_i \rho] = 0$. In this case, the supremum of $i$-th term is $\sup_{\omega_i > 0} -\tr[P_i \sigma] \omega_i = 0$ achieved in the limit $\omega_i \to 0$. Now in the case $\tr[P_i \rho] > 0$ (which is the only possible case when $\rho > 0$), observe that the expression is concave in $\omega_i$, the inner supremum is achieved by the local maxima at $\omega_i = \tr [P_i \rho] / \tr [P_i \sigma]$. Plugging this into~\eqref{eq:supsup1}, we find
\begin{align}
\sup_{ \{ P_i \}_{i=1}^d} \sum_{i=1}^d \tr [P_i \rho] \log \frac{\tr [P_i \rho]}{\tr [P_i \sigma]} + \tr[P_i \rho] - \tr[P_i \rho] = \sup_{ \{ P_i \}_{i=1}^d} \sum_{i=1}^d \tr [P_i \rho] \log \frac{\tr [P_i \rho]}{\tr [P_i \sigma]} \,.
\end{align}
This is the expression for the measured relative entropy in~\eqref{eq:relative_measured}.
The remaining supremum is achieved because the set of rank-1 projectors is compact and the divergence is lower semi-continuous.
    
It remains to show that the two variational expressions in~\eqref{eq:lemma3} are equivalent. We have $\log (1+x) \leq x$ for all $x > -1$ and, thus, $-\log \tr[\sigma \omega] \geq 1 - \tr[\sigma \omega]
$ for all $\omega > 0$. This yields
\begin{align}
\sup_{\omega > 0}\  \tr[\rho \log \omega] - \log \tr[\sigma \omega] \geq \sup_{\omega > 0} \tr[\rho \log \omega] + 1 - \tr[\sigma \omega] \,.
\end{align}
Now note that the expression on the left hand side is invariant under the substitution $\omega \to \gamma \omega$ for $\gamma > 0$. Hence, as $\tr[\sigma \omega] > 0$ for $\omega > 0$ and non-zero $\sigma$, we can add the normalization constraint $\tr[\sigma\omega] = 1$ and we have
\begin{align}
\sup_{\omega > 0}\  \tr[\rho \log \omega] - \log \tr[\sigma \omega] 
&= \sup_{\substack{\omega > 0 \\ \tr[\sigma \omega] = 1}} \tr[\rho \log \omega] - \log \tr[\sigma \omega] \\
&\leq \sup_{\omega > 0} \tr[\rho \log \omega] + 1 - \tr[\sigma \omega] \,,
\end{align}
where we used that $\log \tr[\sigma \omega] = 1 - \tr[\sigma \omega]$ when $\tr[\sigma \omega] = 1$.
\end{proof}

Using this variational expression, we are able to answer a question left open by Donald~\cite[Sec.~3]{donald86} as well as Hiai and Petz~\cite[Sec.~1]{hiai91}, namely whether the two definitions of measured relative entropy are equal. 

\begin{theorem}\label{th:general_rel}
For $\rho \in \cS$ and $\sigma \in \cP$, we have $D^{\Proj}(\rho\|\sigma) = D^{\M}(\rho\|\sigma)$.
\end{theorem}

\begin{proof}
The direction `$\leq$' holds trivially. Moreover, if $\rho \not\ll \sigma$, we can choose $P_1$ to be a rank-$1$ projector such that $\tr[P_1 \rho] > 0$ and $\tr[P_1 \sigma] = 0$ and thus $D^{\Proj}(\rho \| \sigma) = D^{\M}(\rho \| \sigma) = + \infty$.

It remains to show the direction `$\geq$' when $\rho \ll \sigma$ holds. Let $(\cX, M)$ be a POVM. Recall that the distribution $P_{\rho, M}$ is defined by $P_{\rho, M}(x) = \tr[M(x) \rho]$. Introducing $\bar{\cX} = \{ x \in \cX : P_{\rho, M}(x) P_{\sigma, M}(x) > 0\}$, we can write
\begin{align}
D(P_{\rho, M} \| P_{\sigma, M}) &= \sum_{x \in \bar{\cX}} P_{\rho,M}(x) \log \frac{P_{\rho, M}(x)}{P_{\sigma, M}(x)} \\
&= \tr\!\left[ \rho \sum_{x \in \bar{\cX}} M(x) \log \frac{P_{\rho,M}(x)}{P_{\sigma,M}(x)} \right] \\
&= \tr\!\left[ \rho \sum_{x \in \bar{\cX}} \sqrt{M(x)} \log \left( \frac{P_{\rho,M}(x)}{P_{\sigma,M}(x)} \cdot \id \right) \sqrt{M(x)} \right] \,.
\end{align}
Now observe that for any $x \in \bar{\cX}$, the spectrum of the operator $\frac{P_{\rho,M}(x)}{P_{\sigma,M}(x)} \cdot \id$ is included in $(0, \infty)$. As a result, we can apply the operator Jensen inequality for the function $\log$, which is operator concave on $(0,\infty)$ and get
\begin{align}\label{eq:afterjenson1}
D(P_{\rho, M} \| P_{\sigma, M}) &\leq \tr\! \left[ \rho \log \left( \sum_{x \in \bar{\cX}} M(x) \frac{ P_{\rho,M}(x)}{ P_{\sigma,M}(x)} \right) \right] \,.
\end{align}
Now we simply choose 
\begin{align}
\omega = \sum_{x \in \bar{\cX}} M(x) \frac{P_{\rho,M}(x)}{P_{\sigma,M}(x)}
\quad\text{so that}\quad 
\tr[\sigma \omega] = \sum_{x \in \bar{\cX}} P_{\sigma,M}(x) \frac{P_{\rho,M}(x)}{ P_{\sigma,M}(x)} = \sum_{x \in \cX} P_{\rho,M}(x) = 1\,,
\end{align}
and which allows to further bound~\eqref{eq:afterjenson1} by $\tr[\rho \log \omega] + 1 - \tr[\sigma\omega]$. Comparing this with the variational expression for the measured relative entropy in Lemma~\ref{lm:var_rel} yields the desired inequality.
\end{proof}

Hence, the measured relative entropy, $D^{\M}(\rho\|\sigma)$, achieves its maximum for projective rank-$1$ measurements and can be evaluated using the concave optimization problem in Lemma~\ref{lm:var_rel}.

%%%%%%%%%%%%%%%%%%%%%%%%%%%%%%%%%%%%%%%%%%%%

\section{Measured R\'enyi Relative Entropy}\label{sec:renyi_measured}

Here we extend the results of the previous section to R\'enyi divergence. Using the same notation as in the previous section, for $\alpha \in (1,\infty)$ we define the R\'enyi divergence~\cite{renyi61} as
\begin{align}
D_{\alpha}(P\|Q) := \frac{1}{\alpha-1} \log \sum_{x \in \cX} P\big( x\big) \left( \frac{P(x)}{ Q(x)} \right)^{\alpha-1}
\end{align}
if $P \ll Q$ and as $+\infty$ if $P \not\ll Q$. For $\alpha \in (0,1)$ we rewrite the sum as 
\begin{align}\label{eq:verdurenyi}
\sum_{x \in \cX} P(x)^{\alpha} Q(x)^{1-\alpha} \,.
\end{align}
Hence we see that absolute continuity $P \ll Q$ is not necessary to keep $D_{\alpha}$ finite for $\alpha < 1$.  However, the R\'enyi divergence instead diverges to $+\infty$ when $P$ and $Q$ are orthogonal.\footnote{$P$ and $Q$ are orthogonal, denoted $P \perp Q$, if there exists an $A \subseteq \cX$ such that $P(A) = 1$ and $Q(A) = 0$.}  It is well known that the R\'enyi divergence converges to the Kullback-Leibler divergence when $\alpha \to 1$ and we thus set $D \equiv D_1$. Moreover, in the limit $\alpha \to \infty$ we find the max-divergence $D_{\infty}(P\|Q) = \sup_{x \in \cX} \log \frac{P(x)}{Q(x)}$.

Let us now define the \emph{measured R\'enyi relative entropy} as before, namely
\begin{align}
D^{\M}_{\alpha}(\rho\|\sigma) := \sup_{(\cX, M)} D_{\alpha}\big( P_{\rho,M} \big\| P_{\sigma,M} \big) \,.
\end{align}
We will later show that this is equivalent to the following \emph{projectively measured R\'enyi relative entropy}, which we define here for $\alpha \in (1,\infty)$ as
\begin{align}
D^{\Proj}_{\alpha}(\rho\|\sigma) := \frac{1}{\alpha-1} \log Q^{\Proj}_{\alpha}(\rho\|\sigma) \,,
\quad \textrm{with} \quad Q_{\alpha}^{\Proj}(\rho\|\sigma) := \sup_{ \{P_i\}_{i=1}^d } \left\{ \sum_i \tr[P_i \rho]^{\alpha} \tr[P_i \sigma]^{1-\alpha} \right\} \,, \label{eq:q1}
\end{align}
and analogously for $\alpha \in (0,1)$ with
\begin{align}
Q_{\alpha}^{\Proj}(\rho\|\sigma) := \min_{ \{P_i\}_{i=1}^d } \left\{ \sum_i \tr[P_i \rho]^{\alpha} \tr[P_i \sigma]^{1-\alpha} \right\} \label{eq:q2} \,.
\end{align}
Note that the supremum in~\eqref{eq:q1} is achieved and $D^{\Proj}_{\alpha}(\rho\|\sigma)$ is finite whenever $\rho \ll \sigma$. Similarly, the minimum in~\eqref{eq:q2} is non-zero and $D^{\Proj}_{\alpha}(\rho\|\sigma)$ is finite whenever $\rho \not\perp \sigma$, i.e.\ when the two states are not orthogonal. 

Next we give variational expressions for $Q_{\alpha}^{\Proj}(\rho\|\sigma)$ similar to the variational characterization of the measured relative entropy in Lemma~\ref{lm:var_rel}. 

\begin{lemma}\label{lm:var_renyi}
For $\rho \in \cS$ and $\sigma \in \cP$, the following identities hold:
\begin{align}
Q_{\alpha}^{\Proj}(\rho\|\sigma) &= \begin{cases} 
\displaystyle \inf_{\omega > 0} \alpha \tr\big[\rho \omega\big] + (1-\alpha) \tr\Big[ \sigma \omega^{\frac{\alpha}{\alpha-1}} \Big] & \textrm{for $\alpha \in (0, \frac12)$} \\
\displaystyle \inf_{\omega > 0} \alpha \tr\Big[\rho \omega^{1 - \frac{1}{\alpha}}\Big] + (1-\alpha) \tr[ \sigma \omega ] & \textrm{for $\alpha \in [\frac12,1)$} \\
\displaystyle \sup_{\omega > 0} \alpha \tr\Big[\rho \omega^{1 - \frac{1}{\alpha}}\Big] + (1-\alpha) \tr[ \sigma \omega ] & \textrm{for $\alpha \in (1, \infty)$}
\end{cases} \label{eq:notscalable} \\
&= \begin{cases} \label{eq:gen_alberti}
\displaystyle \inf_{\omega > 0} \tr\left[\rho\omega\right]^\alpha\tr\left[\sigma\omega^{\frac{\alpha}{\alpha-1}}\right]^{1-\alpha} & \textrm{for $\alpha \in (0,1)$} \\
\displaystyle \sup_{\omega > 0} \tr\Big[\rho \omega^{1 - \frac{1}{\alpha}}\Big]^\alpha\tr[ \sigma \omega ]^{1-\alpha} & \textrm{for $\alpha \in (1, \infty)$}
\end{cases}
\quad .
\end{align}
The infima and suprema are achieved when $\rho > 0$ and $\sigma > 0$.
\end{lemma}

The expressions \eqref{eq:gen_alberti} can be seen as a generalization of Alberti's theorem~\cite{alberti83} for the fidelity (which corresponds to $\alpha=1/2$) to general $\alpha\in(0,1)\cup(1,\infty)$.

\begin{proof}
We first show the identity~\eqref{eq:notscalable}. Let us discuss the case $\alpha\in (0, 1)$ in detail. Note that the two expressions given for $\alpha \in (0,\frac12)$ and $\alpha \in [\frac12,1)$ are equivalent by the transformation $\omega \mapsto \omega^{\frac{\alpha}{\alpha-1}}$ (the reason for the different ways of writing is to see that the expressions are convex in $\omega$, which will be useful later, in particular in Theorem~\ref{th:general_renyi}). We first write
\begin{align}
\inf_{\omega > 0} \alpha \tr\Big[\rho \omega^{1 - \frac{1}{\alpha}}\Big] + (1-\alpha) \tr[ \sigma \omega ] &=
\inf_{ \{ P_i \}_{i=1}^d} \inf_{\{ \omega_i \}_{i=1}^d} \alpha \sum_i \omega_i^{1-\frac{1}{\alpha}} \tr[P_i \rho] + (1-\alpha) \sum_i \omega_i \tr[P_i \sigma]\,.
\end{align} 
Let $i \in [d]$ be such that $\tr[P_i \rho]$ and $\tr[P_i \sigma]$ are both strictly positive (which is the case when $\rho > 0$ and $\sigma > 0$). Then
a local (and thus global) minimum for $\omega_i$ is easily found at the point where 
\begin{align}
\alpha \Big(1-\frac1{\alpha} \Big) \omega_i^{-\frac{1}{\alpha}} \tr[P_i\rho]  + (1-\alpha) \tr[P_i\sigma] = 0 \quad \implies \quad \omega_i = \left( \frac{\tr[P_i\rho]}{\tr[P_i\sigma]} \right)^{\alpha} \,.
\end{align}
If both $\tr[P_i \rho] = \tr[P_i \sigma] = 0$ we can chose $\omega_i$ arbitrarily.
If only $\tr[P_i \rho] = 0$ the infimum is achieved in the limit $\omega_i \to 0$, and if only $\tr[P_i \sigma] = 0$ in the limit $\omega_i \to \infty$. In all these cases the infimum of the $i$-th term is zero. Furthermore, it is achieved for a finite, non-zero $\omega_i$ when $\rho > 0$ and $\sigma > 0$.
Plugging this solution into the above expression yields
\begin{align}
\inf_{ \{ P_i \}_{i=1}^d} \sum_{i \in [d]} \tr[P_i\rho]^{\alpha} \tr[P_i\sigma]^{1-\alpha}\,.
\end{align}
This infimum is always achieved since the set we optimize over is compact. Comparing this with the definition of $Q_{\alpha}^{\Proj}(\rho\|\sigma)$ yields the first equality. 

For the case $\alpha \in (1,\infty)$, when $\rho \ll \sigma$, the proof is analogous to the previous argument. Otherwise, it is simple to see that the supremum is $+ \infty$.

Now we show~\eqref{eq:gen_alberti}. Using~\eqref{eq:notscalable} and the weighted arithmetic-geometric mean inequality we have
\begin{align}\label{eq:scalable}
Q_{\alpha}^{\Proj}(\rho\|\sigma)\geq\inf_{\omega > 0} \tr\left[\rho\omega^{1-\frac{1}{\alpha}}\right]^\alpha\tr\left[\sigma\omega\right]^{1-\alpha} \quad\textrm{for $\alpha \in (0,1)$}\,.
\end{align}
However, for any feasible $\omega > 0$ in~\eqref{eq:notscalable} and $\lambda > 0$, $\lambda \omega > 0$ is also feasible and choosing $\lambda=\tr\big[\rho\omega^{1-\frac{1}{\alpha}}\big]^\alpha\tr[\sigma\omega]^{-\alpha}$ shows that~\eqref{eq:notscalable} cannot exceed~\eqref{eq:scalable}.
Similarly, by Bernoulli's inequality,
\begin{align}
Q_{\alpha}^{\Proj}(\rho\|\sigma)\leq\sup_{\omega > 0}\tr\left[\rho\omega^{1-\frac{1}{\alpha}}\right]^\alpha\tr\left[\sigma\omega\right]^{1-\alpha}\quad\textrm{for $\alpha \in (1,\infty)$}\,. 
\end{align}
And the same argument as above yields the equality. 
\end{proof}

As for the measured relative entropy, the restriction to rank-$1$ projective measurements is in fact not restrictive at all.

\begin{theorem}\label{th:general_renyi}
For $\rho \in \cS$ and $\sigma \in \cP$, we have $D_{\alpha}^{\Proj}(\rho\|\sigma) = D_{\alpha}^{\M}(\rho\|\sigma)$.
\end{theorem}

\begin{proof}
For $\alpha > 1$ we follow the steps of the proof of Theorem~\ref{th:general_rel}.
Consider any finite set $\cX$ and POVM $M$ with induced measures $P_{\rho,M}$ and $P_{\sigma,M}$. We can write 
\begin{align}\label{eq:alpha_ent_def_proof}
D_{\alpha}(P_{\rho,M}\|P_{\sigma,M}) = \frac{1}{1-\alpha} \log \sum_{x \in \bar{\cX}} P_{\rho,M}(x)^{\alpha} P_{\sigma,M}(x)^{1-\alpha} \,,
\end{align}
where we can restrict the sum over $\bar{\cX} = \{x \in \cX :  P_{\rho,M}(x) P_{\sigma,M}(x) > 0\}$. We then find that the sum satisfies
\begin{align}
\sum_{x \in \bar{\cX}} P_{\rho,M}(x) \left( \frac{ P_{\rho,M}(x)}{ P_{\sigma,M}(x)} \right)^{\alpha-1}
&\leq \tr\!\left[ \rho \left( \sum_{x \in \bar{\cX}} M( x) \left( \frac{ P_{\rho,M} (x)}{ P_{\sigma,M}(x)} \right)^{\alpha} \right)^{1-\frac{1}{\alpha}} \right]\,,
\end{align}
where the inequality again follows by the operator Jensen inequality and the operator concavity of the function $t \mapsto t^{1-\frac{1}{\alpha}}$ on $[0,\infty)$. Now we set
\begin{align}
\omega = \sum_{x \in \bar{\cX}} M( x) \left( \frac{ P_{\rho,M}(x)}{ P_{\sigma,M}(x)} \right)^{\alpha}\quad\text{so that}\quad\tr[\sigma \omega] = \sum_{x \in \cX} P_{\rho,M}( x) \left( \frac{ P_{\rho,M}(x)}{ P_{\sigma,M}(x)} \right)^{\alpha-1}\,.
\end{align}
Thus, we can bound
\begin{align}\label{eq:nicebound1}
\sum_{x\in \cX} P_{\rho,M}( x) \left( \frac{ P_{\rho,M}(x)}{ P_{\sigma,M}(x)} \right)^{\alpha-1} \leq \alpha \tr\!\big[\rho \omega^{1-\frac1\alpha}\big] + (1-\alpha) \tr[\sigma \omega] \,.
\end{align}
Comparing this with the variational expression in Lemma~\ref{lm:var_renyi} yields the desired inequality.

For $\alpha < 1$, we use the same notation as in \eqref{eq:alpha_ent_def_proof}. We further distinguish the cases $\alpha \in (0,\frac12)$ and $\alpha \in [\frac12,1)$. For $\alpha \in (0,\frac12)$, we define
\begin{align}
\omega = \sum_{x\in \bar{\cX}} M(x) \left(\frac{P_{\rho, M}(x)}{ P_{\sigma,M}(x)}\right)^{\alpha-1}\,.
\end{align}
We can then evaluate
\begin{align}
\tr\left[\sigma \omega^{\frac{\alpha}{\alpha-1}}\right] &= \tr\!\left[\sigma \left(\sum_{x \in \bar{\cX}} M( x) \left(\frac{ P_{\rho, M}(x)}{ P_{\sigma,M}(x)}\right)^{\alpha-1} \right)^{\frac{\alpha}{\alpha - 1}}\right] \\
&\leq \tr\!\left[ \sigma \sum_{x \in \bar{\cX}} M( x) \left(\frac{ P_{\rho, M}(x)}{ P_{\sigma,M}(x)} \right)^{\alpha}  \right] = \sum_{x \in \bar{\cX}} P_{\rho, M}(x)^{\alpha} {P}_{\sigma, M}(x)^{1-\alpha}  \,.
\end{align}
where we used the operator convexity of $t \mapsto t^{\frac{\alpha}{\alpha-1}}$ on $(0, \infty)$ and the operator Jensen inequality. Moreover,
\begin{align}
\tr[\rho \omega] = \sum_{x \in \bar{\cX}} P_{\rho, M}(x)^{\alpha} {P}_{\sigma, M}(x)^{1-\alpha} \,.
\end{align}
As a result
\begin{align}
\sum_{x \in \bar{\cX}} P_{\rho,M} (x)^{\alpha} P_{\sigma,M}(x)^{1-\alpha} &\geq \alpha \tr[\rho \omega] + (1-\alpha) \tr[\sigma \omega^{\frac{\alpha}{\alpha-1}}]\,.
\end{align}
Comparing this with the variational expression in Lemma~\ref{lm:var_renyi} yields the desired inequality.

For $\alpha \in [\frac12,1)$ we choose $\omega = \sum_{x \in \bar{\cX}} M(x) \left(\frac{P_{\rho, M}(x)}{P_{\sigma, M}(x)}\right)^{\alpha}$, so that
\begin{align}
\tr\left[\rho \omega^{1 - \frac{1}{\alpha}}\right] &\leq \sum_{x \in \bar{\cX}} P_{\rho, M}(x) \left(\frac{P_{\rho, M}(x)}{ P_{\sigma, M}(x)}\right)^{\alpha - 1} \\
\tr\left[\sigma \omega \right] &=  \sum_{x \in \bar{\cX}} P_{\rho,M} (x)^{\alpha} P_{\sigma,M}(x)^{1-\alpha}  \,,
\end{align}
and once again conclude using the variational expression in Lemma~\ref{lm:var_renyi}.
\end{proof}

%%%%%%%%%%%%%%%%%%%%%%%%%%%%%%%%%%%%%%%%%%%%

\section{Achievability of Relative Entropy}\label{sec:achievable}

\subsection{Umegaki's Relative Entropy}

Here we compare the measured relative entropy to other notions of quantum relative entropy that have been investigated in the literature and have found operational significance in quantum information theory. Umegaki's quantum relative entropy~\cite{umegaki62} has found operational significance as the threshold rate for asymmetric binary quantum hypothesis testing~\cite{hiai91}. For $\rho \in \cS$ and $\sigma \in \cP$, it is defined as
\begin{align}
D(\rho\|\sigma):= \tr\!\big[\rho(\log\rho-\log\sigma)\big]\quad\text{if $\sigma\gg\rho$ and as $+\infty$ if $\sigma \not\gg \rho$.}
\end{align}
We recall the following variational expression by Petz~\cite{petz88} (see also~\cite{kosaki86} for another variational expression):
\begin{align}\label{eq:petz_variational1}
D(\rho\|\sigma) &= \sup_{\omega > 0} \tr[\rho \log \omega] - \log \tr[\exp(\log \sigma +  \log \omega)]\\
&= \sup_{\omega > 0} \tr[\rho \log \omega] + 1 - \tr[\exp(\log \sigma +  \log \omega)]\label{eq:petz_var}\,.
\end{align}
By the data-processing inequality for the relative entropy~\cite{lindblad75,uhlmann77} and Theorem~\ref{th:general_rel} we always have
\begin{align}
D^{\Proj}(\rho\|\sigma)=D^{\M}(\rho\|\sigma)\leq D(\rho\|\sigma)\,,
\end{align}
and moreover Petz~\cite{petz86b} showed the inequality $D^{\Proj}(\rho\|\sigma)\leq D(\rho\|\sigma)$ is strict if $\rho$ and $\sigma$ do not commute (for $\rho>0$ and $\sigma>0)$. Theorem~\ref{th:general_rel} strengthens this to show that the strict inequality persists even when we take a supremum over POVMs. In the following we give an alternative proof of Petz' result and then extend the argument to R\'enyi relative entropy in Section~\ref{sec:sandwhich}.

\begin{prop}\label{prop:petz}
Let $\rho \in \cS$ with $\rho > 0$ and $\sigma\in\cP$ with $\sigma > 0$. Then, we have
\begin{align}\label{eq:petz}
D^{\M}(\rho\|\sigma) \leq D(\rho\|\sigma)\quad\text{with equality if and only if} \quad [\rho,\sigma]=0 \,.
\end{align}
\end{prop}

Our proof relies on the Golden--Thompson inequality~\cite{golden65,thompson65}. It states that for two Hermitian matrices $X$ and $Y$, it holds that
\begin{align}\label{eq:gt}
\tr[\exp(X+Y)]\leq\tr[\exp(X)\exp(Y)]
\end{align}
with equality if an only if $[X,Y] = 0$ as shown in~\cite{so92}.
\begin{proof}
First, it is evident that equality holds if $[\rho,\sigma] = 0$ since then there exists a projective measurement that commutes with $\rho$ and $\sigma$ and thus does not effect the states.
For the following, it is worth writing the variational expressions for the two quantities side by side. Namely, writing $H = \log \omega$, we have
\begin{align}
D(\rho\|\sigma) &= 1 + \sup_{H} \tr[\rho H] - \tr[\exp(\log \sigma +  H)] \\
D^{\M}(\rho\|\sigma) &= 1 + \max_{H} \tr[\rho H] - \tr[\sigma \exp(H)]\,, \label{eq:var2}
\end{align}
where we optimize over all Hermitian operators $H$. Note that, according to Lemma~\ref{lm:var_rel}, we can write a $\max$ for \eqref{eq:var2} because we are assuming $\rho > 0$ and $\sigma > 0$.
The inequality in~\eqref{eq:petz} can now be seen as a direct consequence of the Golden--Thompson inequality.

It remains to show that $D^{\M}(\rho\|\sigma) = D(\rho\|\sigma)$ implies $[\rho,\sigma]= 0$. Let $H^*$ be any maximizer of the variational problem in~\eqref{eq:var2}. Observe now that the equality $D^{\M}(\rho\|\sigma) = D(\rho\|\sigma)$ necessitates
\begin{align}
\tr[\exp(\log \sigma + H^*)] = \tr[\sigma \exp(H^*)]\,,
\end{align}
which holds only if and only if $[\sigma,H^*] = 0$ by the equality condition in~\eqref{eq:gt}. Now define the function
\begin{align}
f(H) = \tr[\rho H] - \tr[\sigma \exp(H)]\,,
\end{align}
and since $H^*$ maximizes $f(H)$, we must have for all Hermitian $\Delta$,
\begin{align}
0 = D f(H^*) [\Delta] = \tr[\rho \Delta] - \tr[\sigma \exp(H^*) \Delta] \,.
\end{align}
To evaluate the second summand of this Fr\'echet derivative we used that $\sigma$ and $H^*$ commute. Since this holds for all $\Delta$ we must in fact have
$\rho = \sigma \exp(H^*)$, which means that $[\rho,\sigma]=0$, as desired.
\end{proof}

In some sense this result tells us that some quantum correlations, as measured by the relative entropy, do not survive the measurement process. This fact appears in quantum information theory in various different guises, for example in the form of locking classical correlations in quantum states~\cite{divincenzo04}. (We also point to~\cite{piani09} for the use of measured relative entropy measures in quantum information theory.) Moreover, since Umegaki's relative entropy is the smallest quantum generalization of the Kullback-Leibler divergence that is both additive and satisfies data-processing (see, e.g.,~\cite[Sec.~4.2.2]{mybook}), the same conclusion can be drawn for any sensible quantum relative entropy. (An example being the quantum relative entropy introduced by Belavkin and Staszewski~\cite{belavkin82}.)

%%%%%%%%%%%%%%%%%%%%%%%%%%%%%%%%%%%%%%%%%%%%

\subsection{Sandwiched R\'enyi Relative Entropy}\label{sec:sandwhich}

Next we consider a family of quantum R\'enyi relative entropies~\cite{lennert13,wilde13} that are commonly called sandwiched R\'enyi relative entropies and have found operational significance since they determine the strong converse exponent in asymmetric binary quantum hypothesis testing~\cite{mosonyiogawa13}. They are of particular interest here because they are, for $\alpha \geq \frac12$, the smallest quantum generalization of the R\'enyi divergence that is both additive and satisfies data-processing~\cite[Sec.~4.2.2]{mybook}. (Examples for other sensible quantum generalizations are the quantum R\'enyi relative entropy first studied by Petz~\cite{petz86} and the quantum divergences introduced by Matsumoto~\cite{matsumoto14}.)

For $\rho \in \cS$ and $\sigma \in \cP$, the sandwiched R\'enyi relative entropy of order $\alpha\in(0,1)\cup(1,\infty)$ is defined as
\begin{align}
{D}_{\alpha}(\rho\|\sigma) := \frac{1}{\alpha-1} \log Q_{\alpha}(\rho\|\sigma)
\quad \textrm{with} \quad 
Q_{\alpha}(\rho\|\sigma) :=\tr\left[\left(\sigma^{\frac{1-\alpha}{2\alpha}}\rho\sigma^{\frac{1-\alpha}{2\alpha}}\right)^\alpha\right] \,,
\end{align}
where the same considerations about finiteness as for the measured R\'enyi relative entropy apply. We also consider the limits $\alpha \to \infty$ and $\alpha \to 1$ of the above expression for which we have~\cite{lennert13},
\begin{align}
D_{\infty}(\rho\|\sigma)=\inf\left\{\lambda\in\mathbb{R}\middle|\rho\leq\exp(\lambda)\sigma\right\}\quad\mathrm{and}\quad D_{1}(\rho\|\sigma) = D(\rho\|\sigma)\,,
\end{align}
respectively. We recall the following variational expression by Frank and Lieb~\cite{frank13}:
\begin{align}
Q_{\alpha}(\rho\|\sigma) &= \begin{cases}
\displaystyle \inf_{\omega > 0} \alpha \tr[\rho\omega] + (1-\alpha) \tr\Big[ \Big( \omega^{\frac12} \sigma^{\frac{\alpha-1}{\alpha}} \omega^{\frac12} \Big)^{\frac{\alpha}{\alpha-1}} \Big] & \textrm{for $\alpha \in (0,1)$}\\
\displaystyle \sup_{\omega > 0} \alpha \tr[\rho\omega] + (1-\alpha) \tr\Big[ \Big( \omega^{\frac12} \sigma^{\frac{\alpha-1}{\alpha}} \omega^{\frac12} \Big)^{\frac{\alpha}{\alpha-1}} \Big] & \textrm{for $\alpha \in (1,\infty)$}\,. \label{eq:frankformula}
\end{cases}
\end{align}
Alternatively, we can also write
\begin{align}\label{eq:newformula}
Q_{\alpha}(\rho\|\sigma) &= \begin{cases}
\displaystyle \inf_{\omega > 0} \tr[\rho\omega]^\alpha\tr\Big[ \Big( \omega^{\frac12} \sigma^{\frac{\alpha-1}{\alpha}} \omega^{\frac12} \Big)^{\frac{\alpha}{\alpha-1}} \Big]^{1-\alpha} & \textrm{for $\alpha \in (0,1)$}\\
\displaystyle \sup_{\omega > 0} \tr[\rho\omega]^\alpha\tr\Big[ \Big( \omega^{\frac12} \sigma^{\frac{\alpha-1}{\alpha}} \omega^{\frac12} \Big)^{\frac{\alpha}{\alpha-1}} \Big]^{1-\alpha} & \textrm{for $\alpha \in (1,\infty)$}\,,
\end{cases}
\end{align}
where we have used the same arguments as in the proof of the second part of Lemma~\ref{lm:var_renyi}. By the data-processing inequality for the sandwiched R\'enyi relative entropy~\cite{lennert13,frank13,beigi13} we always have
\begin{align}
D^{\M}_\alpha(\rho\|\sigma)\leq {D}_\alpha(\rho\|\sigma)\quad\text{for $\alpha\geq\frac12$}\,.
\end{align}
In the following we give an alternative proof of this fact and show that
\begin{align}
[\rho,\sigma]\neq0\quad \implies \quad D^{\M}_\alpha(\rho\|\sigma) < {D}_\alpha(\rho\|\sigma)
\quad \textrm{for} \quad \alpha\in\left(\frac12,\infty\right)\,.
\end{align}
In contrast, at the boundaries $\alpha \in \big\{ \frac12, \infty \big\}$ it is known that $D_{\alpha}(\rho\|\sigma) = D_{\alpha}^{\M}(\rho\|\sigma)$~\cite{koenig08,fuchs96,mosonyiogawa13}. (We refer to~\cite[App.~A]{mosonyiogawa13} for a detailed discussion.)

\begin{theorem}\label{prop:renyi_equality}
Let $\rho \in \cS$ with $\rho > 0$ and $\sigma \in \cP$ with $\sigma > 0$. For $\alpha \in \big(\frac12, \infty\big)$, we have
\begin{align}
D_{\alpha}^{\M}(\rho\|\sigma) \leq {D}_{\alpha}(\rho\|\sigma)\quad\text{with equality if and only if}\quad[\rho,\sigma]=0\,.
\end{align}
\end{theorem}

The argument is similar to the proof of Proposition~\ref{prop:petz} but with the Golden--Thompson inequality replaced by the Araki--Lieb--Thirring inequality~\cite{liebthirring05,araki90}. It states that for $X,Y\geq0$ we have
\begin{align}
\tr\left[YXY\right]^r\leq\tr\left[Y^rX^rY^r\right]& \quad\text{for $r\geq1$}\label{eq:alt_1}\\
\tr\left[YXY\right]^r\geq\tr\left[Y^rX^rY^r\right]& \quad\text{for $r\in[0,1]$}\label{eq:alt_2}\,,
\end{align}
with equality if and only if $[X,Y]=0$ except for $r=1$ as shown in~\cite{hiai94}.

\begin{proof}
We give the proof for $\alpha \in \big(\frac12, 1\big)$ and note that the argument for $\alpha \in (1,\infty)$ is analogous. We have the following variational expressions from Lemma~\ref{lm:var_renyi} and~\eqref{eq:frankformula}:
\begin{align}
Q_\alpha(\rho\|\sigma) &= \inf_{\omega > 0} \alpha \tr[\rho\omega] + (1-\alpha) \tr\Big[ \Big( \omega^{-\frac12} \sigma^{\frac{1-\alpha}{\alpha}} \omega^{-\frac12} \Big)^{\frac{\alpha}{1-\alpha}} \Big]\label{eq:Qalpha}\\
Q^{\Proj}_\alpha(\rho\|\sigma) &= \min_{\omega > 0} \alpha \tr[\rho\omega] + (1-\alpha) \tr\left[ \sigma \omega^{\frac{\alpha}{\alpha-1}} \right]\label{eq:Qalpha_meas}\,,
\end{align}
where the existence of the minima relies on the fact that both operators have full support.
(Note also that these two expressions are in fact equivalent for $\alpha = \frac12$.)
Since $\frac{\alpha}{1-\alpha}\geq1$, the inequality then follows immediately by the Araki--Lieb--Thirring inequality~\eqref{eq:alt_1}:
\begin{align}
Q_\alpha^{\Proj}(\rho\|\sigma)\geq Q_\alpha(\rho\|\sigma)\quad
\implies
\quad D_\alpha^{\M}(\rho\|\sigma)\leq  {D}_{\alpha}(\rho\|\sigma)\,.
\end{align}
Furthermore, if $[\rho,\sigma]=0$ we have equality. To show that $Q^{\Proj}_\alpha(\rho\|\sigma) = Q_\alpha(\rho\|\sigma)$ implies $[\rho,\sigma]= 0$, we define the function
\begin{align}
f_\alpha(\omega) = \alpha \tr[\rho\omega] + (1-\alpha) \tr\left[ \sigma \omega^{\frac{\alpha}{\alpha-1}} \right]\,. \label{eq:theminhere}
\end{align}
For $\omega^*_\alpha$ any minimizer of the variational problem in~\eqref{eq:Qalpha_meas}, we have
\begin{align}
0 = D f_\alpha(\omega^*_\alpha) [\Delta] = \alpha\tr[\rho \Delta]-\alpha\tr\left[\sigma\left(\omega^*_\alpha\right)^{\frac{1}{\alpha-1}}\Delta\right]\,,
\end{align}
for all Hermitian $\Delta$. To evaluate the second summand of this Fr\'echet derivative we used that $\sigma$ and $\omega^*_\alpha$ commute, which holds by the equality condition for Araki--Lieb--Thirring. We thus conclude that
$\rho=\sigma\left(\omega^*_\alpha\right)^{\frac{1}{\alpha-1}}$ which implies that $[\rho,\sigma]=0$.
\end{proof}

%%%%%%%%%%%%%%%%%%%%%%%%%%%%%%%%%%%%%%%%%%%%

\section{Violation of Data-Processing for $\alpha < \frac12$}

As a further application of the variational characterization of measured R\'enyi relative entropy, we can show that the data-processing for the sandwiched R\'enyi relative entropy $D_{\alpha}$ fails for $\alpha < \frac12$. (Numerical evidence pointed to the fact that data-processing does not hold in this regime~\cite{martinthesis}.)

\begin{theorem}
Let $\rho \in \cS$ with $\rho > 0$ and $\sigma\in\cP$ with $\sigma > 0$, and $[\rho,\sigma] \neq 0$. For $\alpha \in\big(0, \frac12\big)$, we have $D_{\alpha}^{\Proj}(\rho\|\sigma) > {D}_{\alpha}(\rho\|\sigma)$.
\end{theorem}

In particular, there exists a rank-$1$ measurement that achieves $D_{\alpha}^{\Proj}(\rho\|\sigma)$ and thus violates the data-processing inequality.

\begin{proof}
First note that $[\rho,\sigma] \neq 0$ implies that the two states are not orthogonal and thus both quantities are finite.
For $\alpha \in (0, \frac12)$ the formulas~\eqref{eq:Qalpha} and~\eqref{eq:Qalpha_meas} still hold. However, in contrast to the proof of Theorem~\ref{prop:renyi_equality} we have $\frac{\alpha}{1-\alpha}\in[0,1]$. Hence, we find by the Araki--Lieb--Thirring inequality~\eqref{eq:alt_2} that
\begin{align}
Q_\alpha^{\Proj}(\rho\|\sigma)\leq Q_\alpha(\rho\|\sigma)\quad \implies \quad D_\alpha^{\Proj}(\rho\|\sigma)\geq{D}_{\alpha}(\rho\|\sigma)\,.
\end{align}
As in the proof of Theorem~\ref{prop:renyi_equality} we have equality if and only if $[\rho,\sigma]=0$. This implies the claim.
\end{proof}

%%%%%%%%%%%%%%%%%%%%%%%%%%%%%%%%%%%%%%%%%%%%

\section{Exploiting Variational Formulas}\label{sec:recovery}

\subsection{Some Optimization Problems in Quantum Information}\label{sec:exploiting}

The variational characterizations of the relative entropy~\eqref{eq:petz_variational1}--\eqref{eq:petz_var}, the sandwiched R\'enyi relative entropy~\eqref{eq:frankformula}--\eqref{eq:newformula}, and their measured counterparts (Lemma~\ref{lm:var_rel} and Lemma~\ref{lm:var_renyi}), can be used to derive properties of various entropic quantities that appear in quantum information theory. We are interested in operational quantities of the form
\begin{align}\label{eq:m_quantity}
\mathcal{M}(\rho):=\min_{\sigma\in\mathcal{C}}\mathbb{D}(\rho\|\sigma)\,,
\end{align}
where $\mathbb{D}(\cdot\|\cdot)$ stands for any relative entropy, measured relative entropy, or R\'enyi variant thereof, and $\mathcal{C}$ denotes some convex, compact set of states. For Umegaki's relative entropy $\mathbb{D} \equiv D$, prominent examples for $\mathcal{C}$ include the set of
\begin{itemize}
\setlength\itemsep{0.1em}
\item separable states, giving rise to the relative entropy of entanglement~\cite{vedral98}.
\item positive partial transpose states, leading to the Rains bound on entanglement distillation~\cite{rains01}.
\item non-distillable states, leading to bounds on entanglement distillation~\cite{vedral99}.
\item quantum Markov states, leading to insights about the robustness properties of these states~\cite{linden08}.
\item locally recoverable states, leading to bounds on the quantum conditional mutual information~\cite{fawzirenner14,seshadreesan14,brandao14}.
\item $k$-extendible states, leading to bounds on squashed entanglement~\cite{liwinter14}.
\end{itemize}
Other examples are conditional R\'enyi entropies which are defined by optimizing the sandwiched R\'enyi relative entropy over a convex set of product states with a fixed marginal, see, e.g., \cite{tomamichel13}.

A central question is what properties of the underlying relative entropy $\mathbb{D}$ translate to properties of the induced measure $\mathcal{M}(\cdot)$. For example, all the relative entropies discussed in this paper are superadditive on tensor product states in the sense that 
\begin{align}\label{eq:additivity_rel_ent}
\mathbb{D}(\rho_1\otimes\rho_2\|\sigma_1\otimes\sigma_2) \geq
\mathbb{D}(\rho_1\|\sigma_1) + \mathbb{D}(\rho_2\|\sigma_2)\,.
\end{align}
We might then ask if we also have
\begin{align}
\min_{\sigma_{12}\in\mathcal{C}_{12}}\mathbb{D}(\rho_1 \otimes \rho_2 \| \sigma_{12}) = \mathcal{M}(\rho_1 \otimes\rho_2)\stackrel{?}{\geq}\mathcal{M}(\rho_1)+\mathcal{M}(\rho_2) = \min_{\sigma_1\in\mathcal{C}_{1}}\mathbb{D}(\rho_1 \| \sigma_1) + \min_{\sigma_2\in\mathcal{C}_{2}}\mathbb{D}( \rho_2 \| \sigma_2) \,. \label{eq:addquestion}
\end{align}
To study questions like this we propose to make use of the variational characterizations of the form
\begin{align}
\mathbb{D}(\rho\|\sigma)=\sup_{\omega>0}f(\rho,\sigma,\omega)\quad\text{in order to write}\quad\mathcal{M}(\rho)=\min_{\sigma\in\mathcal{C}}\sup_{\omega>0}f(\rho,\sigma,\omega)=\sup_{\omega>0}\min_{\sigma\in\mathcal{C}}f(\rho,\sigma,\omega)\,,
\end{align}
where we made use of Sion's minimax theorem~\cite{sion58} for the last equality. We note that the conditions of the minimax theorem are often fulfilled. The minimization over $\sigma\in\mathcal{C}$ then typically simplifies and is a convex or even semidefinite optimization. (As an example, for the measured relative entropies the objective function becomes linear in $\sigma$.) We can then use strong duality of convex optimization to rewrite this minimization as a maximization problem~\cite{boyd04}:
\begin{align}
\min_{\sigma\in\mathcal{C}}f(\rho,\sigma,\omega)=\max_{\bar{\sigma}\in\bar{\mathcal{C}}}\bar{f}(\rho,\bar{\sigma},\omega)\,.
\end{align}
This leads to the expression
\begin{align}\label{eq:dualm_quantity}
\mathcal{M}(\rho)=\sup_{\omega>0}\max_{\bar{\sigma}\in\bar{\mathcal{C}}}\bar{f}(\rho,\bar{\sigma},\omega)\,,
\end{align}
which, in contrast to the definition of $\mathcal{M}(\rho)$ in~\eqref{eq:m_quantity}, only involves maximizations. This often gives useful insights about $\mathcal{M}(\rho)$. As an example, let us come back to the question of superadditivity raised in~\eqref{eq:addquestion}. We want to argue that the following two conditions on $\bar{f}$ and $\bar{\mathcal{C}}$ imply superadditivity. First, we need that the function $\bar{f}$ is superadditive itself, i.e.\ we require that
\begin{align}
\bar{f}(\rho_1 \otimes \rho_2, \bar{\sigma}_1 \otimes \bar{\sigma}_2, \omega_1 \otimes \omega_2 ) \geq \bar{f}(\rho_1, \bar{\sigma}_1, \omega_1) + \bar{f}(\rho_2, \bar{\sigma}_2, \omega_2) \,.
\end{align}
And second, we require that the sets $\bar{\mathcal{C}}$ are closed under tensor products in the sense that $\bar{\sigma}_1 \in \bar{\mathcal{C}}_1$ and $\bar{\sigma}_2 \in \bar{\mathcal{C}}_2$ imply that $\bar{\sigma}_1 \otimes \bar{\sigma}_2 \in \bar{\mathcal{C}}_{12}$.
Using these two properties, we deduce that
\begin{align}
\mathcal{M}(\rho_1 \otimes \rho_2) \geq \bar{f}(\rho_1 \otimes \rho_2, \bar{\sigma}_1 \otimes \bar{\sigma}_2, \omega_1 \otimes \omega_2 )\geq \bar{f}(\rho_1, \bar{\sigma}_1, \omega_1) + \bar{f}(\rho_2, \bar{\sigma}_2, \omega_2)
\end{align}
for any $\omega_1, \omega_2 > 0$ and any $\bar{\sigma}_1 \in \bar{\mathcal{C}}_1$, $\bar{\sigma}_2 \in \bar{\mathcal{C}}_2$. Hence, the inequalities also hold true if we maximize over these variables, implying superadditivity.

%%%%%%%%%%%%%%%%%%%%%%%%%%%%%%%%%%%%%%%%%%%%

\subsection{Relative Entropy of Recovery}

In the following we denote multipartite quantum systems using capital letters, e.g., $A$, $B$, $C$. The set of density operators on $A$ and $B$ is then denoted $\cS(AB)$, for example. We also use these letters as subscripts to indicate which systems operators act on.

As a concrete application, we study the additivity properties of the \emph{relative entropy of recovery} defined as~\cite{seshadreesan14,brandao14,hirche15}
\begin{align}\label{eq:rel_ent_rec}
D^{\rec}(\rho_{AD}\|\sigma_{AE}):=\inf_{\Gamma_{E\to D}}D(\rho_{AD}\|(\cI_{A}\otimes\Gamma_{E\to D})(\sigma_{AE}))\,,
\end{align}
where the infimum is taken over all trace-preserving completely positive maps $\Gamma_{E \to D}$. (We restrict to $\sigma_A>0$ such that the quantity is surely finite and the infimum becomes a minimum.) 
One motivation for studying the additivity properties of the relative entropy of recovery is the study of lower bounds on the quantum conditional mutual information and strengthenings of the data-processing inequality for the relative entropy~\cite{liwinter14,fawzirenner14,seshadreesan14,brandao14,berta15,sutter15,renner15,wilde15,junge15}. In particular, \cite[Prop.~3]{brandao14} shows that
\begin{align}\label{eq:cqmi_lower}
I(A:C|B)\geq\lim_{n\to\infty}\frac{1}{n}D^{\rec}\left(\rho_{ABC}^{\otimes n}\|\rho_{AB}^{\otimes n}\right)\,,
\end{align}
where the systems are understood as $D = BC$ and $E = B$. To obtain a lower bound that does not involve limits, the authors of~\cite{brandao14} use a Finetti-type theorem to show that 
\begin{align}\label{eq:second_step_bhos}
\lim_{n \to \infty} \frac{1}{n}D^{\rec}(\rho_{ABC}^{\otimes n} \| \rho_{AB}^{\otimes n}) \geq D^{\M, \rec}(\rho_{ABC} \| \rho_{AB})\,.
\end{align}
with the \emph{measured relative entropy of recovery} defined as~\cite{brandao14} (see also~\cite[Rmk.~6]{seshadreesan14}),
\begin{align}\label{eq:rel_ent_rec2}
D^{\M,\rec}(\rho_{AD}\|\sigma_{AE})&:=\inf_{\Gamma_{E\to D}}D^{\M}(\rho_{AD} \|(\cI_{A}\otimes\Gamma_{E\to D})(\sigma_{AE}))\,.
\end{align}
(We restrict to $\sigma_A>0$ such that the quantity is surely finite and the infimum becomes a minimum.) This gives an interpretation for the conditional mutual information in terms of recoverability.

The measured relative entropy is superadditive and if this property would translate to $D^{\M,\rec}$ we would get an alternative proof for the step~\eqref{eq:second_step_bhos}. Using the variational expression for the measured relative entropy (Lemma~\ref{lm:var_rel}) together with strong duality for semidefinite programming, we find the following alternative characterization for $D^{\M,\rec}$.

\begin{lemma}\label{lem:recovery_measured}
Let $\rho_{AD}\in\cS(AD)$ and $\sigma_{AE}\in\cS(AE)$ with $\sigma_A>0$, and let $\sigma_{AEF}$ be a purification of $\sigma_{AE}$. Then, we have
\begin{align}\label{eq:measured_dual}
\begin{aligned}
D^{\M, \rec}(\rho_{AD} \| \sigma_{AE})=\;&\textnormal{maximize}
& & \tr[\rho_{AD} \log R_{AD}]  \\
& \textnormal{subject to}
& &S_{AF}>0,\;R_{AD}>0\\
& & &\id_D\otimes S_{AF}\geq R_{AD}\otimes\id_F\\
& & &\tr[S_{AF}\sigma_{AF}]=1\,.
\end{aligned}
\end{align}
\end{lemma}

\begin{proof}
We write
\begin{align}
\Gamma_{E\to D}(\sigma_{AE})= \tr_{F}\big[\Gamma_{E\to D}(\sigma_{AEF}) \big]= \tr_{F} \Big[\Gamma_{E \to D} \big( \sqrt{\sigma_{AF}} \Psi_{AF:E} \sqrt{\sigma_{AF}} \big) \Big]\,,
\end{align}
where we denote by $\Psi_{AF:E}$ the (unnormalized) maximally entangled state between $AF:E$ in the basis of the Schmidt decomposition of $\sigma_{AEF}$ with respect to $AF:E$. With this we define the Choi-Jamio{\l}kowski state (unnormalized) of the map $\Gamma_{E\to D}$ as
\begin{align}\label{eq:choistate}
\tau_{ADF} = \Gamma_{E \to D} \big( \Psi_{AF:E} \big),\quad\tau_{AF}=\Pi_{AF}^{\sigma}\,,
\end{align}
where $\Pi_{AF}^{\sigma}$ denotes the projector onto the support of $\sigma_{AF}$. Hence, we can write
\begin{align}\label{eq:choi_state}
\Gamma_{E\to D}(\sigma_{AE}) = \tr_F\big[\sqrt{\sigma_{AF}} \tau_{ADF} \sqrt{\sigma_{AF}}\big]\,,
\end{align}
and thus we can express the optimization problem for $\Gamma_{E \to D}$ in terms of the Choi-Jamio{\l}kowski state in~\eqref{eq:choistate}. Together with the variational characterization of the measured relative entropy (Lemma~\ref{lm:var_rel}) we find
\begin{align}
D^{\M, \rec}(\rho_{AD} \| \sigma_{AE}) = \min_{\tau_{ADF} \in \mathrm{Rec}(\sigma_{AE})} \sup_{R_{AD} > 0}  \tr[\rho_{AD} \log R_{AD}] + 1 - \tr\left[\tau_{ADF}\sqrt{\sigma_{AF}}R_{AD}\sqrt{\sigma_{AF}}\right]\,,
\end{align}
where $\mathrm{Rec}(\sigma_{AE}):= \{\tau_{ADF}\geq0,\;\tau_{AF}=\Pi_{AF}^{\sigma}\}$. We now apply Sion's minimax theorem~\cite{sion58} to exchange the minimum with the supremum. The theorem is applicable as $\mathrm{Rec}(\sigma_{AE})$ is convex and compact and $\{ \omega_{AD} > 0\}$ is convex. Moreover, as the logarithm is operator concave, the function
\begin{align}
R_{AD} \mapsto \tr[\rho_{AD} \log R_{AD}] + 1 - \tr\left[\tau_{ADF}\sqrt{\sigma_{AF}}R_{AD}\sqrt{\sigma_{AF}}\right]
\end{align}
is concave for any fixed $\tau_{ADF}$. Finally, for any fixed $R_{AD}$, the function being optimized is linear on $\tau_{ADF}$. As a result,
\begin{align}
D^{\M, \rec}(\rho_{AD} \| \sigma_{AE})=\sup_{R_{AD} > 0}\tr[\rho_{AD} \log R_{AD}] + 1 - \max_{\tau_{ADF} \in \mathrm{Rec}(\sigma_{AE})}\tr\left[\tau_{ADF}\sqrt{\sigma_{AF}}R_{AD}\sqrt{\sigma_{AF}}\right]\,.
\end{align}
Now for a fixed $R_{AD}>0$, the inner maximization is a semidefinite program for which we can write the following programs:
\begin{align}\label{eq:convexdual}
\begin{matrix}
\begin{array}{rl}
\text{maximize:}\quad & \tr\left[\tau_{ADF}\sqrt{\sigma_{AF}}R_{AD}\sqrt{\sigma_{AF}}\right] \\
\text{subject to:}\quad &  \tau_{ADF}\geq0 \\
&\tau_{AF}=\Pi_{AF}^{\sigma}
\end{array}
&\qquad
\begin{array}{rl}
\text{minimize:}\quad & \tr[S_{AF}\sigma_{AF}] \\
\text{subject to:}\quad & \id_D\otimes S_{AF}\geq R_{AD}\otimes\id_F\,.\\~
\end{array}  
\end{matrix}
\end{align}
As the primal problem is strictly feasible, Slater's condition (see, e.g., \cite{boyd04}) is satisfied and we have strong duality. This leads to the expression:
\begin{align}
\begin{aligned}
D^{\M, \rec}(\rho_{AD} \| \sigma_{AE})=\;&\textnormal{maximize}
& & \tr[\rho_{AD} \log R_{AD}]+1-\tr[S_{AF}\sigma_{AF}]\\
& \textnormal{subject to}
& &S_{AF}>0,\;R_{AD}>0\\
& & &\id_D\otimes S_{AF}\geq R_{AD}\otimes\id_F\,.
\end{aligned}
\end{align}
Observe now that we can restrict the optimization to $\tr[S_{AF}\sigma_{AF}] = 1$. In fact, for arbitrary feasible $R_{AD}$ and $S_{AF}$, we can define
\begin{align}
\tilde{R}_{AD} = \frac{R_{AD}}{\tr[S_{AF}\sigma_{AF}]}\quad\mathrm{and}\quad\tilde{S}_{AF} = \frac{S_{AF}}{\tr[S_{AF}\sigma_{AF}]}\,.
\end{align}
The constraint $\id_D\otimes S_{AF}\geq R_{AD}\otimes\id_F$ is still satisfied and the value of the objective function can only increase,
\begin{align}
\tr[\rho_{AD} \log \tilde{R}_{AD}] &= \tr[\rho_{AD} \log R_{AD}] - \tr[\rho_{AD}] \log \tr[S_{AF}\sigma_{AF}]\label{eq:increase1}\\
 & \geq \tr[\rho_{AD} \log R_{AD}] - \tr[S_{AF}\sigma_{AF}]+1\,,\label{eq:increase2}
\end{align}
where we used the inequality $\log x \leq x - 1$ and that $\tr[\rho_{AD}] = 1$. This concludes the proof.
\end{proof}

This readily implies that $D^{\M,\rec}$ is indeed superadditive.

\begin{prop}\label{thm:super_additive_rel_ent_rec}
Let $\rho_{AD}\in\cS(AD)$, $\tau_{A'D'}\in\cS(A'D')$, $\sigma_{AE}\in\cS(AE)$, and $\omega_{A'E'}\in\cS(A'E')$ with $\sigma_A,\omega_{A'}>0$. Then, we have
\begin{align}
D^{\M,\rec}(\rho_{AD}\otimes\tau_{A'D'}\|\sigma_{AE}\otimes\omega_{A'E'})\geq D^{\M,\rec}(\rho_{AD}\|\sigma_{AE})+D^{\M,\rec}(\tau_{A'D'}\|\omega_{A'E'})\,.
\end{align}
\end{prop}

\begin{proof}
Given feasible operators $R_{AD},S_{AF}$ for the quantity $D^{\M,\rec}(\rho_{AD}\|\sigma_{AE})$ and feasible operators $R_{A'D'},S_{A'F'}$ for the quantity $D^{\M,\rec}(\tau_{A'D'}\|\omega_{A'E'})$, we have
\begin{align}
\id_D\otimes S_{AF}\geq R_{AD}\otimes\id_F \ &\land \ \id_{D'}\otimes S_{A'F'}\geq R_{A'D'}\otimes\id_{F'} \nonumber\\
&\implies\nonumber\\
\id_{DD'}\otimes S_{AF}\otimes S_{A'F'}&\geq R_{AD}\otimes R_{A'D'}\otimes\id_{FF'}\,.\label{eq:feasible}
\end{align}
Here we used that $A \geq B \implies M \otimes A \geq M \otimes B$ for $M \geq 0$, which holds since taking the tensor product with $M$ is a positive map. Moreover, we have $\tr\big[(S_{AF}\otimes S_{A'F'})(\sigma_{AF}\otimes\sigma_{A'F'})\big]=1$. In other words, $(R_{AD}\otimes R_{A'D'},S_{AF}\otimes S_{A'F'})$ is a feasible pair in the expression~\eqref{eq:measured_dual} for $D^{\M,\rec}(\rho_{AD}\otimes\tau_{A'D'}\|\sigma_{AE}\otimes\omega_{A'E'})$ and we get
\begin{align}
D^{\M,\rec}(\rho_{AD}\otimes\tau_{A'D'}\|\sigma_{AE}\otimes\omega_{A'E'})
&\geq \tr\left[\left(\rho_{AD}\otimes\tau_{A'D'}\right) \log \left(R_{AD}\otimes R_{A'D'}\right)\right] \\
&= \tr\left[\rho_{AD}\log R_{AD}\right] + \tr\left[\tau_{A'D'}\log R_{A'D'}\right]\,.
\end{align}
Taking the supremum over feasible $(R_{AD},S_{AF})$ and $(R_{A'D'},S_{A'F'})$, we get the claimed superadditivity.
\end{proof}

Now, if the relative entropy of recovery would be superadditive we could strengthen~\eqref{eq:second_step_bhos} to
\begin{align}
\lim_{n \to \infty} \frac{1}{n}D^{\rec}(\rho_{ABC}^{\otimes n} \| \rho_{AB}^{\otimes n})\stackrel{?}{\geq}D^{\rec}(\rho_{ABC} \| \rho_{AB})\,.
\end{align}
This together with~\eqref{eq:cqmi_lower} would lead to a stronger lower bound on the conditional mutual information. Using strong duality for convex programming we can write $D^{\rec}$ in a dual form similar to Lemma~\ref{lem:recovery_measured}.

\begin{lemma}\label{lem:entropy_of_recovery}
Let $\rho_{AD}\in\cS(AD)$ and $\sigma_{AE}\in\cS(AE)$ with $\sigma_A>0$, and let $\sigma_{AEF}$ be a purification of $\sigma_{AE}$. Then, we have
\begin{align}
\begin{aligned}
D^{\rec}(\rho_{AD} \| \sigma_{AE})=\;&\textnormal{maximize}
& & \tr[\rho_{AD} \log \rho_{AD}]-D^{\M}(\rho_{AD}\|R_{AD}) \\
& \textnormal{subject to}
& &S_{AF}>0,\;R_{AD}>0\\
& & &\id_D\otimes S_{AF}\geq R_{AD}\otimes\id_F\\
& & &\tr[S_{AF}\sigma_{AF}]=1\,,
\end{aligned}
\end{align}
where $R_{AD}$ is not normalized and the measured relative entropy term is in general negative.
\end{lemma}

We provide a proof in Appendix~\ref{app:missing}. This is to be compared to our expression for the measured relative entropy of recovery from Lemma~\ref{lem:recovery_measured}, which can be written as
\begin{align}
\begin{aligned}
D^{\M, \rec}(\rho_{AD} \| \sigma_{AE})=\;&\textnormal{maximize}
& & \tr[\rho_{AD} \log \rho_{AD}]-D(\rho_{AD}\|R_{AD}) \\
& \textnormal{subject to}
& &S_{AF}>0,\;R_{AD}>0\\
& & &\id_D\otimes S_{AF}\geq R_{AD}\otimes\id_F\\
& & &\tr[S_{AF}\sigma_{AF}]=1\,,
\end{aligned}
\end{align}
where $R_{AD}$ is not normalized and the relative entropy term is in general negative. Unfortunately, we can not solve the open additivity question for the relative entropy of recovery using the variational expression from Lemma~\ref{lem:entropy_of_recovery}. However, we note that the argument flipped relative entropy of recovery
\begin{align}
\bar{D}^{\rec}(\sigma_{AE}\|\rho_{AD}):=\inf_{\Gamma_{E \to D}}D((\mathcal{I}_A\otimes\Gamma_{E\to D})(\sigma_{AE})\|\rho_{AD})
\end{align}
becomes additive (restricted to $\rho_A>0$ to assure finiteness and making the infimum a minimum).

\begin{prop}\label{lem:interchanged}
Let $\rho_{AD}\in\cS(AD)$ and $\sigma_{AE}\in\cS(AE)$ with $\rho_A>0$, and let $\sigma_{AEF}$ be a purification of $\sigma_{AE}$. Then, we have
\begin{align}\label{eq:interchanged1}
\begin{aligned}
\bar{D}^{\rec}(\sigma_{AE}\|\rho_{AD})=\;&\textnormal{maximize}
& & -\tr[S_{AF}\sigma_{AF}]\\
& \textnormal{subject to}
& &R_{AD}>0\\
& & &\id_D\otimes S_{AF}\geq(\log\rho_{AD}-\log R_{AD})\otimes\id_F\\
& & &\tr[R_{AD}]=1\,.
\end{aligned}
\end{align}
Moreover, for $\tau_{A'D'}\in\cS(A'D')$ with $\tau_{A'}>0$ and $\omega_{A'E'}\in\cS(A'E')$, we have
\begin{align}\label{eq:interchanged2}
\bar{D}^{\rec}(\sigma_{AE}\otimes\omega_{A'E'}\|\rho_{AD}\otimes\tau_{A'D'})=\bar{D}^{\rec}(\sigma_{AE}\|\rho_{AD})+\bar{D}^{\rec}(\omega_{A'E'}\|\tau_{A'D'})\,.
\end{align}
\end{prop}

We provide a proof based on strong duality for convex programming in Appendix~\ref{app:missing}.

%%%%%%%%%%%%%%%%%%%%%%%%%%%%%%%%%%%%%%%%%%%%

\subsection{R\'enyi Relative Entropy of Recovery}\label{sec:renyi_recovery}

As a generalization of the measured relative entropy of recovery~\eqref{eq:rel_ent_rec2} we define the \emph{measured R\'enyi relative entropy of recovery} for $\alpha \in (1,\infty)$ as (see also~\cite[Rmk.~6]{seshadreesan14}),
\begin{align}
&D^{\M,\rec}_\alpha(\rho_{AD}\|\sigma_{AE}) := \frac{1}{\alpha-1} \log Q^{\M,\rec}_{\alpha}(\rho_{AD}\|\sigma_{AE}) \,,\\
&\textrm{with} \quad
Q^{\M,\rec}_{\alpha}(\rho_{AD}\|\sigma_{AE}) := \inf_{\Gamma_{E\to D}} Q^{\M}_{\alpha}(\rho_{AD}\|(\cI_{A}\otimes\Gamma_{E\to D})(\sigma_{AE}))\,,
\end{align}
and analogously for $\alpha \in (0,1)$ with
\begin{align}
Q^{\M,\rec}_{\alpha}(\rho_{AD}\|\sigma_{AE}) := \sup_{\Gamma_{E\to D}} Q^{\M}_{\alpha}(\rho_{AD}\|(\cI_{A}\otimes\Gamma_{E\to D})(\sigma_{AE}))\,,
\end{align}
where we restrict to $\sigma_A>0$ such that the quantity is surely finite and the infimum/supremum is achieved. Using the variational expression for the measured R\'enyi relative entropy (Lemma~\ref{lm:var_renyi}) together with strong duality for semidefinite programming, we find the following alternative characterization for $Q^{\M,\rec}_\alpha$.

\begin{lemma}\label{lem:renyimeas_dual}
Let $\rho_{AD}\in\cS(AD)$ and $\sigma_{AE}\in\cS(AE)$ with $\sigma_A>0$, and let $\sigma_{AEF}$ be a purification of $\sigma_{AE}$. For $\alpha \in [\frac12,1)$, we have
\begin{align}
\begin{aligned}
Q^{\M,\rec}_\alpha(\rho_{AD}\|\sigma_{AE})=\;&\textnormal{minimize}
& & \tr\left[\rho_{AD}R_{AD}^{1-\frac{1}{\alpha}}\right]^\alpha\tr[S_{AF}\sigma_{AF}]^{1-\alpha}\\
& \textnormal{subject to}
& &S_{AF}>0,\;R_{AD}>0\\
& & &\id_D\otimes S_{AF}\geq R_{AD}\otimes\id_F\,.
\end{aligned}
\end{align}
Similar dual formulas hold for $\alpha \in (0,\frac12)\cup(1,\infty)$.
\end{lemma}

\begin{proof}
As in the proof of Lemma~\ref{lem:recovery_measured} and using the first variational formula from Lemma~\ref{lm:var_renyi}, we get from Sion's minimax theorem~\cite{sion58} that
\begin{align}
Q_\alpha^{\M,\rec}(\rho_{AD}\|\sigma_{AE})&=\max_{\tau_{ADF}\in\mathrm{Rec}(\sigma_{AE})}\inf_{R_{AD}>0}\alpha\tr\Big[\rho_{AD}R_{AD}^{1-\frac{1}{\alpha}}\Big]+(1-\alpha)\tr\left[\tau_{ADF}\sqrt{\sigma_{AF}}R_{AD}\sqrt{\sigma_{AF}}\right]\\
&=\inf_{R_{AD}>0}\max_{\tau_{ADF}\in\mathrm{Rec}(\sigma_{AE})}\alpha\tr\Big[\rho_{AD}R_{AD}^{1-\frac{1}{\alpha}}\Big]+(1-\alpha)\tr\left[\tau_{ADF}\sqrt{\sigma_{AF}}R_{AD}\sqrt{\sigma_{AF}}\right]\,.
\end{align}
The set $\mathrm{Rec}(\sigma_{AE})$ is convex and compact and $\{R_{AD}>0\}$ is convex. Moreover, the function
\begin{align}
R_{AD}\mapsto\alpha\tr\Big[\rho_{AD}R_{AD}^{1-\frac{1}{\alpha}}\Big]+(1-\alpha)\tr\left[\tau_{ADF}\sqrt{\sigma_{AF}}R_{AD}\sqrt{\sigma_{AF}}\right]
\end{align}
is convex for any fixed $\tau_{ADF}$ because of the operator convexity of $t \mapsto t^{\beta}$ with $\beta = 1 - \frac{1}{\alpha}$ for $\alpha \in [\frac12, 1)$ (with $\beta\in[-1,0)$). Finally, for a fixed $R_{AD}$, the function being optimized is linear on $\tau_{ADF}$. As in~\eqref{eq:convexdual} we then look at the convex dual of
\begin{align}
\max_{\tau_{ADF}\in\mathrm{Rec}(\sigma_{AE})}\tr\left[\tau_{ADF}\sqrt{\sigma_{AF}}R_{AD}\sqrt{\sigma_{AF}}\right]\,,
\end{align}
and find
\begin{align}
\begin{aligned}
Q^{\M,\rec}_\alpha(\rho_{AD}\|\sigma_{AE})=\;&\textnormal{minimize}
& & \alpha\tr\left[\rho_{AD}R_{AD}^{1-\frac{1}{\alpha}}\right]+(1-\alpha)\tr[S_{AF}\sigma_{AF}]\\
& \textnormal{subject to}
& &S_{AF}>0,\;R_{AD}>0\\
& & &\id_D\otimes S_{AF}\geq R_{AD}\otimes\id_F\,.
\end{aligned}
\end{align}
Invoking the arithmetic-geometric mean inequality as in the proof of Lemma~\ref{lm:var_renyi} establishes the claim.
\end{proof}

We find that the measured R\'enyi entropy of recovery is superadditive for all R\'enyi parameters.

\begin{prop}\label{thm:super_additive_rel_ent_rec_renyi}
Let $\rho_{AD}\in\cS(AD)$, $\tau_{A'D'}\in\cS(A'D')$, $\sigma_{AE}\in\cS(AE)$, and $\omega_{A'E'}\in\cS(A'E')$ with $\sigma_A,\omega_{A'}>0$. For $\alpha\in(0,1)\cup(1,\infty)$, we have
\begin{align}
D_{\alpha}^{\M, \rec}(\rho_{AD}\otimes\tau_{A'D'}\|\sigma_{AE}\otimes\omega_{A'E'})\geq D_{\alpha}^{\M, \rec}(\rho_{AD}\|\sigma_{AE})+D_{\alpha}^{\M, \rec}(\tau_{A'D'}\|\omega_{A'E'})\,.
\end{align}
\end{prop}

\begin{proof}
We give the argument for $\alpha\in[\frac12,1)$ based on Lemma~\ref{lem:renyimeas_dual} and note that the proof for $\alpha\in (0,\frac12)\cup(1,\infty)$ is similar. Given feasible operators $R_{AD},S_{AF}$ for the quantity $Q^{\M,\rec}_\alpha(\rho_{AD}\|\sigma_{AE})$ and feasible operators $R_{A'D'},S_{A'F'}$ for the quantity $Q^{\M,\rec}_\alpha(\tau_{A'D'}\|\omega_{A'E'})$, we apply exactly the same argument as in~\eqref{eq:feasible} and find
\begin{align}
Q_{\alpha}^{\M, \rec}(\rho_{AD}\otimes\tau_{A'D'}\|\sigma_{AE}\otimes\omega_{A'E'})\leq Q_{\alpha}^{\M, \rec}(\rho_{AD}\|\sigma_{AE}) \cdot Q_{\alpha}^{\M, \rec}(\tau_{A'D'}\|\omega_{A'E'})\,.
\end{align}
This concludes the proof.
\end{proof}

Analogously, the relative entropy of recovery~\eqref{eq:rel_ent_rec} is generalized to the \emph{R\'enyi relative entropy of recovery} as (see also~\cite[Rmk.~6]{seshadreesan14}),
\begin{align}
&D^{\rec}_\alpha(\rho_{AD}\|\sigma_{AE}):=\frac{1}{\alpha-1} \log Q^{\rec}_{\alpha}(\rho_{AD}\|\sigma_{AE})
\end{align}
with
\begin{align}
Q^{\rec}_{\alpha}(\rho_{AD}\|\sigma_{AE}):=
\begin{cases}
\displaystyle\inf_{\Gamma_{E\to D}} Q_{\alpha}(\rho_{AD}\|(\cI_{A}\otimes\Gamma_{E\to D})(\sigma_{AE}))\quad&\text{for $\alpha \in (1,\infty)$}\\
\displaystyle\sup_{\Gamma_{E\to D}} Q_{\alpha}(\rho_{AD}\|(\cI_{A}\otimes\Gamma_{E\to D})(\sigma_{AE}))&\text{for $\alpha \in (0,1)$}\,,
\end{cases}
\end{align}
where we restrict to $\sigma_A>0$ such that the quantity is surely finite and the infimum/supremum is achieved. Now, for the special case of $\alpha = 1/2$, the sandwiched R\'enyi relative entropy is equal to the negative logarithm of the fidelity and as the measured fidelity is the same as the fidelity (see, e.g., \cite[Sec.~3.3]{fuchs96}), we have
\begin{align}
Q_{1/2}^{\M}=Q_{1/2}\quad\text{as well as}\quad\mathrm{Q}^{\rec}_{1/2}=\mathrm{Q}^{\M,\rec}_{1/2}\,.
\end{align}
The same holds true at $\alpha\to\infty$~\cite{mosonyiogawa13}. As such Proposition~\ref{thm:super_additive_rel_ent_rec_renyi} can be seen as a generalization of the multiplicativity of the fidelity of recovery~\cite{seshadreesan14} that was derived by two of the authors~\cite{berta15}. For the general sandwiched R\'enyi relative entropy of recovery, we only have preliminary numerics indicating that they are not additive. Using a non-commutative extension of the techniques from~\cite[Sec.~3.3.1]{ben-tal01}, we can write the sandwiched R\'enyi relative entropy with fractional R\'enyi parameter as a semidefinite program. Together with the Choi state as in~\eqref{eq:choi_state}, we then get a semidefinite program for, e.g., $D^{\rec}_{2/3}$. Using this, we found weak numerical evidence for additivity violations for small dimensional systems. This also challenges the possible additivity of the relative entropy of recovery $D^{\rec}$ (corresponding to the case $\alpha=1$).\footnote{A very recent preprint~\cite{fawzi17} has demonstrated by numerical examples that the relative entropy of recovery is indeed non-additive.}

%%%%%%%%%%%%%%%%%%%%%%%%%%%%%%%%%%%%%%%%%%%%

\section{Conclusion}\label{sec:conclusion}

We presented variational characterizations of the measured relative entropy and the measured R\'enyi relative entropy. Using these formulas we were able to show that these quantities can be achieved by rank-1 projective measurements instead of general POVMs. We also showed that the measured R\'enyi relative entropy is equal to the corresponding sandwiched R\'enyi relative entropy if and only if the two quantum states commute (except for the special cases with R\'enyi parameters $1/2$ and $\infty$ where they are always equal), and gave analytical counterexamples for the data-processing inequality for the sandwiched R\'enyi relative entropy with R\'enyi parameters smaller than $1/2$. Finally, we applied our variational characterizations to analyze the additivity properties of various relative entropies of recovery.

As an extension of our work it would be desirable to weaken the support conditions in Proposition~\ref{prop:petz} and Theorem~\ref{prop:renyi_equality} concerning the equality conditions for measured relative entropy and relative entropy. (We note that this is non-trivial because we make use of the equality conditions for Golden-Thompson and Araki--Lieb--Thirring). Moreover, it would be neat to extend the concept of measured relative entropy to general, continuous POVMs described by measure spaces (see, e.g., \cite{vanerven14} for the definition of  Kullback-Leibler divergence and R\'enyi divergence on measure spaces). Many of the steps still go through and we would like to point to a Jensen inequality for operator-valued measures~\cite{farenick07}. It seems, however, that an extension of this inequality would be needed~\cite{farenickpc15}. Finally, it would also be interesting to use the variational characterizations presented in this work for studying the operational entropic quantities mentioned in Section~\ref{sec:exploiting}. This might lead to new applications of measured relative entropy in quantum information theory.\\

%%%%%%%%%%%%%%%%%%%%%%%%%%%%%%%%%%%%%%%%%%%%

\paragraph*{Acknowledgments:}

We acknowledge discussions with Fernando Brand\~{a}o, Douglas Farenick and Hamza Fawzi. MB acknowledges funding by the SNSF through a fellowship, funding by the Institute for Quantum Information and Matter (IQIM), an NSF Physics Frontiers Center (NFS Grant PHY-1125565) with support of the Gordon and Betty Moore Foundation (GBMF-12500028), and funding support form the ARO grant for Research on Quantum Algorithms at the IQIM (W911NF-12-1-0521). Most of this work was done while OF was also with the Department of Computing and Mathematical Sciences, California Institute of Technology. MT would like to thank the IQIM at CalTech and John Preskill for his hospitality during the time most of the technical aspects of this project were completed. He is funded by an ARC Discovery Early Career Researcher Award fellowship (Grant No. DE160100821).

%%%%%%%%%%%%%%%%%%%%%%%%%%%%%%%%%%%%%%%%%%%%

\appendix

\section{Proofs for Results in Section~\ref{sec:recovery}}\label{app:missing}

\begin{proof}[Proof of Lemma~\ref{lem:entropy_of_recovery}]
We start by writing out the relative entropy of recovery as a convex optimization program,
\begin{align}
\begin{aligned}
D^{\rec}(\rho_{AD} \| \sigma_{AE})=\tr[\rho_{AD}\log\rho_{AD}]+\;&\textnormal{minimize}
& & -\tr[\rho_{AD} \log\gamma_{AD}]\\
& \textnormal{subject to}
& &\gamma_{AD}=\tr_F\left[\sqrt{\sigma_{AF}}\tau_{ADF}\sqrt{\sigma_{AF}}\right]\\
& & &\tau_{ADF}\geq0\\
& & &\tau_{AF}=\Pi^\sigma_{AF}\,,
\end{aligned}
\end{align}
where the notation is as in the proof of Lemma~\ref{lem:recovery_measured}. Clearly, the first and last constraint can be relaxed to $\gamma_{AD} \leq \tr_F\left[\sqrt{\sigma_{AF}}\tau_{ADF}\sqrt{\sigma_{AF}}\right]$ and $\tau_{AF} \leq \Pi^\sigma_{AF}$ without changing the solution. 
The Lagrangian can then be written as
\begin{align}
\mathcal{L}(\gamma, \tau, R, S) &= - \tr[\rho_{AD} \log \gamma_{AD}]+\tr\left[R_{AD}\left(\gamma_{AD}-\tr_F\left[\sqrt{\sigma_{AF}}\tau_{ADF}\sqrt{\sigma_{AF}}\right]\right)\right]\notag\\
&\quad +\tr\left[S_{AF}\left(\tau_{AF}-\Pi^\sigma_{AF}\right)\right]\\
&= - \tr[\rho_{AD} \log \gamma_{AD}]+\tr[R_{AD}\gamma_{AD}]-\tr\left[\tau_{ADF}\left(\sqrt{\sigma_{AF}}R_{AD}\sqrt{\sigma_{AF}}-S_{AF}\right)\right]\notag\\
&\quad-\tr\left[S_{AF}\Pi^\sigma_{AF}\right]\, ,
\end{align}
where we introduced the variables $R_{AD} \geq 0$ and $S_{AF} \geq 0$.
In order to compute the dual objective function, we should take the infimum of this quantity over $\gamma_{AD}>0$ and $\tau_{ADF} \geq 0$. Using the variational expression for the measured relative entropy (Lemma~\ref{lm:var_rel}) we find
\begin{align}
\inf_{\gamma_{AD}>0,\,\tau_{ADF} \geq 0} \mathcal{L}(\gamma, \theta, R, S) =-D^{\M}(\rho_{AD}\|R_{AD})+1-\tr[S_{AF}\Pi^\sigma_{AF}]\,,
\end{align}
when $S_{AF}\geq\sqrt{\sigma_{AF}}R_{AD}\sqrt{\sigma_{AF}}$. With Slater's strong duality and a change of variable $\bar{S}_{AF}:=\sigma_{AF}^{-1/2}S_{AF}\sigma_{AF}^{-1/2}$ (but not including the bar in the following), we get
\begin{align}
\begin{aligned}
D^{\rec}(\rho_{AD} \| \sigma_{AE})=\tr[\rho_{AD}\log\rho_{AD}]+\;&\textnormal{maximize}
& & -D^{\M}(\rho_{AD}\|R_{AD})+1-\tr[S_{AF}\sigma_{AF}]\\
& \textnormal{subject to}
& &S_{AF}>0,\;R_{AD}>0\\
& & &\id_D\otimes S_{AF}\geq R_{AD}\otimes\id_F\,.
\end{aligned}
\end{align}
Adding the constraint $\tr[S_{AF}\sigma_{AF}]=1$ as in the proof of Lemma~\ref{lem:recovery_measured} concludes the proof.
\end{proof}

\begin{proof}[Proof of Proposition~\ref{lem:interchanged}]
We start by writing out the argument flipped relative entropy of recovery as a convex optimization program,
\begin{align}
\begin{aligned}
\bar{D}^{\rec}(\sigma_{AE}\|\rho_{AD})=\;&\textnormal{minimize}
& & \tr[\gamma_{AD} \log\gamma_{AD}]-\tr[\gamma_{AD} \log\rho_{AD}]\\
& \textnormal{subject to}
& &\gamma_{AD}=\tr_F\left[\sqrt{\sigma_{AF}}\tau_{ADF}\sqrt{\sigma_{AF}}\right]\\
& & &\tau_{ADF}\geq0\\
& & &\tau_{AF}=\Pi^\sigma_{AF}\,,
\end{aligned}
\end{align}
where the notation is as in the proof of Lemma~\ref{lem:recovery_measured}. The Lagrangian can be written as
\begin{align}
\mathcal{L}(\gamma, \tau, T, S) &= \tr[\gamma_{AD} \log \gamma_{AD}]-\tr[\gamma_{AD} \log \rho_{AD}]+\tr\left[T_{AD}\left(\gamma_{AD}-\tr_F\left[\sqrt{\sigma_{AF}}\tau_{ADF}\sqrt{\sigma_{AF}}\right]\right)\right]\notag\\
&\quad+\tr\left[S_{AF}\left(\tau_{AF}-\Pi^\sigma_{AF}\right)\right]\\
&= \tr[\gamma_{AD} \log \gamma_{AD}]+\tr[\gamma_{AD}(T_{AD}-\log\rho_{AD})]\notag\\
&\quad-\tr\left[S_{AF}\Pi^\sigma_{AF}\right]-\tr\left[\tau_{ADF}\left(\sqrt{\sigma_{AF}}T_{AD}\sqrt{\sigma_{AF}}-S_{AF}\right)\right]\,,
\end{align}
where $T_{AD}$ and $S_{AF}$ are Hermitian operators.
In order to compute the dual objective function, we should take the infimum of this quantity over $\gamma_{AD}>0$ and $\tau_{ADF} \geq 0$. From the last expression we get $S_{AF}\geq\sqrt{\sigma_{AF}}T_{AD}\sqrt{\sigma_{AF}}$ and we also know how to optimize the first expression as it is an entropy maximization question:
\begin{align}
\inf_{\gamma_{AD}>0} \tr[\gamma_{AD} \log \gamma_{AD}]+\tr[\gamma_{AD}(T_{AD}-\log\rho_{AD})]\,.
\end{align}
It is optimized when $\gamma_{AD} = \exp(\log \rho_{AD} - T_{AD} + \alpha \id_{AD})$ for some $\alpha$~\cite{chandrasekaran15}. This means that this infimum is given by
\begin{align}
&\inf_{\alpha} \tr[\exp(\log \rho_{AD} - T_{AD} + \alpha \id_{AD}) \left( \log \rho_{AD} - T_{AD} + \alpha \id_{AD} \right)]\notag\\
&\quad+ \tr[\exp(\log \rho_{AD} - T_{AD} + \alpha \id_{AD}) \left( T_{AD} - \log \rho_{AD} \right)]\notag\\
&= \inf_{\alpha} \alpha \tr[\exp(\log \rho_{AD} - T_{AD} + \alpha \id_{AD})] \\
&= \inf_{\alpha} \alpha e^{\alpha} \tr[\exp(\log \rho_{AD} - T_{AD})] \\
&= - \tr[\exp(\log \rho_{AD} -T_{AD} - \id_{AD})]\,.
\end{align}
With Slater's strong duality and a change of variable $\bar{S}_{AF}:=\sigma_{AF}^{-1/2}S_{AF}\sigma_{AF}^{-1/2}$ (but not including the bar in the following), we get
\begin{align}
\begin{aligned}
\bar{D}^{\rec}(\sigma_{AE}\|\rho_{AD})=\;&\textnormal{maximize}
& & -\tr[S_{AF}\sigma_{AF}]-\tr[\exp(\log \rho_{AD} -T_{AD} - \id_{AD})]\\
& \textnormal{subject to}
& &\id_D\otimes S_{AF}\geq T_{AD}\otimes\id_F\,.
\end{aligned}
\end{align}
We now do another change of variable (but not including the bar in the following)
\begin{align}
R_{AD}:= \exp(\log \rho_{AD} - T_{AD} - \id_{AD})\quad\text{as well as}\quad\id_{D}\otimes\bar{S}_{AF}:=\id_{D}\otimes S_{AF}+\id_{ADF}\,,
\end{align}
and the program becomes
\begin{align}
\begin{aligned}
\bar{D}^{\rec}(\sigma_{AE}\|\rho_{AD})=\;&\textnormal{maximize}
& & -\tr[S_{AF}\sigma_{AF}]+1-\tr[R_{AD}]\\
& \textnormal{subject to}
& &R_{AD}>0\\
& & &\id_D\otimes S_{AF}\geq(\log\rho_{AD}-\log R_{AD})\otimes\id_F\,.
\end{aligned}
\end{align}
Observe now that we can add the constraint $\tr[R_{AD}] = 1$. In fact, let
\begin{align}
\bar{R}_{AD}:= \frac{R_{AD}}{\tr[R_{AD}]}\quad\mathrm{and}\quad\bar{S}_{AF}:= S_{AF} + \log \tr[R_{AD}]\,.
\end{align}
This solution satisfies the constraint and the objective value becomes
\begin{align}
-\tr[\bar{S}_{AF}\sigma_{AF}]+1- \tr[\bar{R}_{AD}]= - \tr[S_{AF}\sigma_{AF}] - \log \tr[R_{AD}] \geq - \tr[S_{AF}\sigma_{AF}] + 1 - \tr[R_{AD}]\,.
\end{align}
This concludes the proof of~\eqref{eq:interchanged1}.

To prove~\eqref{eq:interchanged2} first note that it is immediate from the definition of the argument flipped relative entropy of recovery that
\begin{align}
\bar{D}^{\rec}(\sigma_{AE}\otimes\omega_{A'E'}\|\rho_{AD}\otimes\tau_{A'D'})\leq\bar{D}^{\rec}(\sigma_{AE}\|\rho_{AD})+\bar{D}^{\rec}(\omega_{A'E'}\|\tau_{A'D'})\,,
\end{align}
and in the following we prove inequality in the other direction using the dual representation~\eqref{eq:interchanged1}. Given feasible operators $R_{AD},S_{AF}$ for the quantity $\bar{D}^{\rec}(\sigma_{AE}\|\rho_{AD})$ and feasible operators $R_{A'D'},S_{A'F'}$ for the quantity $\bar{D}^{\rec}(\omega_{A'E'}\|\tau_{A'D'})$, we have
\begin{align}
&\id_D\otimes S_{AF}\geq(\log\rho_{AD}-\log R_{AD})\otimes\id_F \land \ \id_{D'}\otimes S_{A'F'}\geq(\log\tau_{A'D'}-\log R_{A'D'})\otimes\id_{F'} \nonumber\\
&\implies\id_{DD'}\otimes(S_{AF}\otimes \id_{A'F'}+\id_{AF}\otimes S_{A'F'})\nonumber\\
&\quad\quad\;\;\geq\big((\log\rho_{AD}-\log R_{AD})\otimes\id_{A'D'}+\id_{AD}\otimes(\log\tau_{A'D'}-\log R_{A'D'})\big)\otimes\id_{FF'}\nonumber\\
&\quad\quad\;\;=\big(\log(\rho_{AD}\otimes\tau_{A'D'})-\log(R_{AD}\otimes R_{A'D'})\big)\otimes\id_{FF'}\,,
\end{align}
just by multiplying with identities and adding the resulting operator inequalities. Moreover, we have $\tr[R_{AD}\otimes R_{A'D'}]=1$. Hence, $(R_{AD}\otimes R_{A'D'},S_{AF}\otimes\id_{A'F'}+\id_{AF}\otimes S_{A'F'})$ is a feasible pair in the expression~\eqref{eq:interchanged1} for $\bar{D}^{\rec}(\sigma_{AE}\otimes\omega_{A'E'}\|\rho_{AD}\otimes\tau_{A'D'})$ and we get
\begin{align}
\bar{D}^{\rec}(\sigma_{AE}\otimes\omega_{A'E'}\|\rho_{AD}\otimes\tau_{A'D'})
&\geq-\tr[(S_{AF}\otimes\id_{A'F'}+\id_{AF}\otimes S_{A'F'})(\sigma_{AF}\otimes\omega_{A'F'})]\\
&=-\tr[S_{AF}\sigma_{AF}]-\tr[S_{A'F'}\omega_{A'F'}]\,.
\end{align}
Taking the supremum over feasible $(R_{AD},S_{AF})$ and $(R_{A'D'},S_{A'F'})$, we find the claimed additivity.
\end{proof}

%%%%%%%%%%%%%%%%%%%%%%%%%%%%%%%%%%%%%%%%%%%%

\bibliographystyle{ultimate}
\bibliography{library}

\end{document}